\documentclass[usenames,dvipsnames,letterpaper,11pt]{article} 
\usepackage{setspace}
\usepackage{array}
\usepackage{enumitem}
\usepackage{amsmath}
\usepackage[T1]{fontenc}
\usepackage[margin=1in]{geometry}
\usepackage[numbers, semicolon]{natbib}
\setcitestyle{authoryear,open={(},close={)},aysep={}} 
\usepackage{amsfonts}
\usepackage{float}
\usepackage{graphicx}
\usepackage{tikz}
\usetikzlibrary{arrows,shapes,snakes,automata,backgrounds,petri,fit, positioning,patterns}
\usepackage{mathrsfs}
\usepackage{caption}
\usepackage{subcaption}
\usepackage[titles,subfigure]{tocloft}
\usepackage{layouts}
\usepackage{hyperref}
\usepackage[normalem]{ulem}
\usepackage{authblk}
\usepackage{bm}
\usepackage{hyperref}
\DeclareMathOperator*{\argmax}{arg\,max}
\newcommand{\simiid}{\stackrel{\textnormal{i.i.d.}}{\sim}}
\usepackage[ruled,procnumbered]{algorithm2e}
\usepackage{amsthm}
\usepackage{lscape}
\allowdisplaybreaks

\newtheorem{theorem}{Theorem}

\newtheorem{corollary}{Corollary}
\newtheorem{lemma}{Lemma}

\theoremstyle{remark}

\newcommand{\p}{\mathbb P}

\newcommand{\di}[1]{\mathop{d#1}}

\DeclareMathOperator{\Var}{Var}

\DeclareMathOperator{\Unif}{Unif}

\DeclareMathOperator{\Bern}{Bern}

\DeclareMathOperator{\TV}{TV}

\DeclareMathOperator{\ETV}{ETV}
\DeclareMathOperator{\HETV}{HETV}

\newcommand{\e}{\mathbb E}

\newcommand{\ci}{\mathrel{\perp\mspace{-10mu}\perp}}

\newcommand{\cip}{\stackrel{p}{\to}}

\usepackage{etoolbox}

\makeatletter
\AtBeginEnvironment{procedure}{\let\c@algocf\c@procedure}
\makeatother

\title{Total Variation Floodgate\\ for Variable Importance Inference in Classification}
\author{Wenshuo Wang\footnote{wenshuowang1997@gmail.com, corresponding author}, Lucas Janson\footnote{Harvard University, ljanson@fas.harvard.edu}, Lihua Lei\footnote{Stanford University, lihualei@stanford.edu} ~and Aaditya Ramdas\footnote{Carnegie Mellon University, aramdas@cmu.edu}}
\date{\today}

\begin{document}

\maketitle

\begin{abstract}
Inferring variable importance is the key problem of many scientific studies, where researchers seek to learn the effect of a feature $X$ on the outcome $Y$ in the presence of confounding variables $Z$. Focusing on classification problems, we define the expected total variation (ETV), which is an intuitive and deterministic measure of variable importance that does not rely on any model context. We then introduce algorithms for  statistical inference on the ETV under design-based/model-X assumptions. These algorithms build on the floodgate notion for regression problems \citep{LZ-LJ:2020}. The algorithms we introduce can leverage any user-specified regression function and produce asymptotic  lower confidence bounds for the ETV. We show the effectiveness of our algorithms with simulations and a case study in conjoint analysis on the US general election.
\end{abstract}


\section{Introduction}
\subsection{Motivation}
In many scientific studies, researchers would like to understand the effect of a feature $X$ on a response variable $Y$, while controlling for potential confounding features $Z$. While this question is sometimes simplified to a hypothesis testing problem of ``does $X$ affect $Y$ at all in the presence of $Z$'', it is more desirable to follow up with ``if so, by how much''; that is, we wish to provide a quantitative variable importance measure (VIM). In traditional statistical frameworks, such a follow-up question is addressed by postulating a parametric model of $\mathcal L(Y\mid X,Z)$ and looking at the inferred parameters. However, such parametric models are often limited in their capacity to capture complex relationships.
This paper aims at defining a VIM for classification problems and conducting inference on it. As a design objective, this (population-level) VIM must be model-free, in that it does not rely on an underlying model assumptions. We would also like the VIM to be intuitive and easy to interpret.

\subsection{Our contribution}

The main contributions of this work are listed below.
\begin{enumerate}
    \item We propose the expected total variation (ETV) as a VIM that is well-defined for any type of variables $(X, Y, Z)$. In this paper, we focus on categorical response variables $Y$, for which VIMs with rigorous statistical guarantees were rarely discussed in the literature and ETV has both an intuitive model-free interpretation and sound statistical properties. 
    \item We introduce algorithms that make statistical inference on the ETV. These algorithms work without imposing any assumptions on the distribution $\mathcal L(Y\mid X,Z)$. Instead, we make the design-based/model-X assumption that we can sample from $\mathcal L(X\mid Z)$, which we discuss in Secion~\ref{sec:fg-etv}. We accompany our algorithms with practical parameter choice recommendations and show that they work well in our extensive simulation results.
    \item We demonstrate the effectiveness of our algorithms in a real conjoint data analysis study.
\end{enumerate}

In the remainder of this section, we will discuss related work and introduce notation. The mathematical definition of ETV will be given in Section~\ref{sec:floodgate-categorical}. We will discuss the properties of ETV in Section~\ref{sec:etv}. Section~\ref{sec:fg-etv} is devoted to our main floodgate algorithm to conduction inference on ETV. We then study the algorithm parameters in Section~\ref{sec:c} to facilitate practical applications and discuss a generalization of ETV in Section~\ref{sec:hierarchical}. Section~\ref{sec:sim} includes simulations on synthetic data to support the effectiveness of our proposed method. We then apply our method to a conjoint analysis example on political candidate preferences in Secton~\ref{sec:conjoint}. 

\subsection{Related work}
The canonical VIM is defined through parametric models. When $Y$ is categorical, the textbook approach is to parameterize $\mathcal L(Y\mid X,Z)$ with a generalized linear model \citep{agresti2015foundations}, and there has been work on parameter inference in high-dimensional sparse models \citep{van2014asymptotically,belloni2016post}. Generalized linear models have limited capacity in capturing non-linear effects, and such parameter-based VIMs crucially rely on the model specification and become ill-defined when the model is misspecified.

A more contemporary line of work utilizes machine learning methods to capture variable importance \citep{fisher2019all,watson2021testing,molnar2023model} in a model-free manner. While these VIMs are well-defined without parametric assumptions, they are associated with a trained machine learning model, which depends on the model choice and the data itself. 

Another existing approach \citep{castro2009polynomial,williamson2020efficient,ning2022shapley} borrows ideas from the game theory literature and considers VIMs based on the Shapley value \citep{shapley1953value}. Shapley-value based VIMs capture the variable's predictive power, and are generally positive even for statistically null variables (that is, $X$ where $X\ci Y\mid Z$). Thus, while these VIMs have attractive predictive interpretations, they lack a causal interpretation.

\citet{azadkia2021simple} introduced a model-free VIM based on cumulative distributions functions for non-categorical $Y$; \citet{huang2020kernel} generalized it to categorical $Y$, but it relies on a user-specified kernel function. These VIMs have the appealing property that they are $0$ if and only if $Y\ci X\mid Z$ and $1$ if and only if $Y$ is a deterministic measurable function of $(X,Z)$. Both papers considered consistent estimators of the VIMs and not lower confidence bounds. \citet{LZ-LJ:2020} proposed a model-free VIM called the minimum mean squared error (mMSE) for non-categorical $Y$ and provided a lower confidence bound. \citet[Section~3.1]{LZ-LJ:2020} extended their inference to a VIM called the mean absolute conditional mean (MACM), which is defined for binary $Y$. Similarly, \citet{williamson2021general} defined a class of VIMs based on a variable's additional predictiveness and constructed estimators and confidence intervals. These VIMs hinge on a pre-determined family of predictors, and while the VIMs could work for any type of responses, \citet{williamson2021general} also only considered inference in cases of non-categorical and binary responses where the conditional mean is meaningful.

To the best of our knowledge, our work is the first to propose a model-free VIM for general categorical responses that is model-free, natural and easy to interpret, and provide a lower confidence bound for it.

Finally, there are two methods in the literature \citep{LZ-LJ:2020,michel2022high} that have particularly strong connections with our proposed method. As it requires first introducing our algorithm to properly discuss comparisons with these methods, we will review these works and their relationship to ours in Section~\ref{sec:comp}.



\subsection{Notation}
For random variables or vectors $W_1$ and $W_2$, $\mathcal L(W_1)$ means the distribution of $W_1$ and $\mathcal L(W_1\,|\,W_2)$ means the conditional distribution of $W_1$ given $W_2$. $\TV(\mathcal L_1,\mathcal L_2)$ means the total variation distance between distributions $\mathcal L_1$ and $\mathcal L_2$. Unless otherwise specified, vectors are column vectors. $\Phi$ denotes the cumulative distribution function of the standard Gaussian distribution $\mathcal N(0,1)$.
\section{Floodgate for categorical response}
\label{sec:floodgate-categorical}
\subsection{The expected total variation distance}
\label{sec:etv}
Let $(X,Y,Z)$ be a random vector with three components. We would like to quantify the effect size of $X$; that is, the strength of the conditional dependence between $X$ and $Y$ given $Z$. We  propose to use the expected total variation distance (ETV) between $\mathcal L(Y\mid X,Z)$ and $\mathcal L(Y\mid Z)$, defined as
\begin{equation}
\ETV(X,Y,Z)=\frac{1}{1-1/|\mathcal Y|}\e[\TV(\mathcal L(Y\mid X,Z),\mathcal L(Y\mid Z))],
\label{equation:fg_target}
\end{equation}
as the VIM, where $|\mathcal Y|$ is the support size of $Y$ (if $|\mathcal Y|=\infty$, we simply normalize by $1$). The normalizing factor is to ensure the value of ETV to be in $[0,1]$, as stated in Lemma~\ref{lemma:etv-range}. 
\begin{lemma}[Range of ETV]
If $\p(Y\in\mathcal Y)=1$, then $0 \leq \ETV(X,Y,Z)\le1$, where the equality is achieved when $X$, conditional on $Z$, deterministically determines $Y$ and $\p(Y=y\mid Z)=1/|\mathcal Y|$ for all $y\in\mathcal Y$.
\label{lemma:etv-range}
\end{lemma}
Note that if $X$ is continuous and $\p(Y=y\mid Z)=1/|\mathcal Y|$ for all $y\in\mathcal Y$, then there always exists $\mathcal L(X\mid Z)$ such that $Y$ is a deterministic function of $(X, Z)$. This means that for any support $\mathcal Y$, there exists $(X,Y,Z)$ such that $\ETV(X,Y,Z)$ is equal to $1$ for finite $|\mathcal Y|$, or gets arbitrarily close to $1$ for $|\mathcal Y|=\infty$.

The ETV has several desirable properties that distinguish it from other VIMs in the literature: 
\begin{enumerate}
    \item ETV  captures all possible effects of $X$ on $Y$ conditional on $Z$, including both linear and non-linear effects.
    \item ETV has a very simple and intuitive form and does not depend on any pre-specified model or kernel function (that is, it is model-free).
    \item $\ETV$ attains its minimum value zero if and only if $X\ci Y\mid Z$.
    \item Particularly ideal for categorical $Y$, the value of ETV does not change under one-to-one transformation of $Y$.
\end{enumerate}
ETV's simple form should already make it very easy to interpret conceptually. The key component is the total variation distance, which is also the Wasserstein distance with 0-1 loss and half the the $L_1$ distance between probability density or mass functions. To visually demonstrate the ETV, we consider an example where $Y$ could be one of 5 categories $s_1,\dots,s_5$. In Figure~\ref{figure:etv}, the red bars represent the probability of $Y$ being from each category conditional on $(X,Z)$ taking on some particular value $(x,z)$, the blue bars represent the probability of $Y$ taking each value conditional on $Z=z$, and the yellow bars represent the difference. The expected sum of the yellow bars, when averaged over $(X,Z)$, is equal to $2(1-1/|\mathcal Y|)\ETV$.
\begin{figure}[H]
    \centering
\includegraphics[width = 0.9\textwidth]{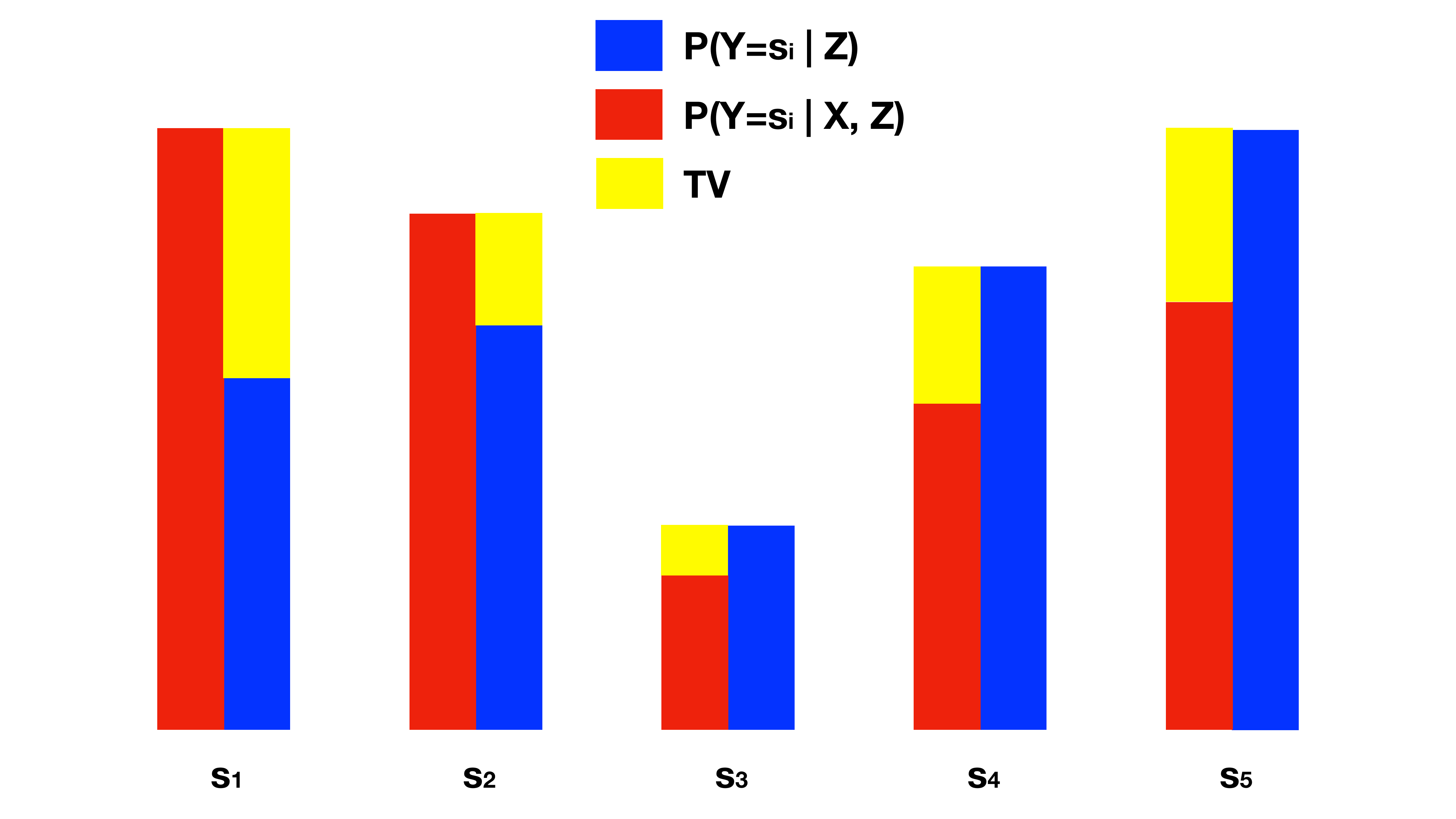}
    \caption{ETV illustration. The expected sum of the yellow bars, when averaged over $(X,Z)$, is equal to $2(1-1/|\mathcal Y|)\ETV$.}
    \label{figure:etv}
\end{figure}



Finally, we briefly discuss the interpretation of ETV in the context of sensitivity analysis. When $Y\ci X\mid Z$ does not hold, we can hypothesize the existence of a confounding variable $U$, such that $Y\ci X \mid Z,U$. We can define $B(X,Z,U)\ge1$ as the (almost sure) supremum of 
\begin{equation*}
\left\{\frac{p(X \mid U, Z)}{p(X \mid Z)},\frac{p(X \mid Z)}{p(X \mid U, Z)}\right\}.
\end{equation*}
Therefore, $B$ measures the minimal confounding effect that can explain away the conditional non-independence. We can show that $B(X,Z,U)\ge1+2(1-1/|\mathcal Y|)\ETV(X,Y,Z)$ (see Appendix~\ref{sec:sens}). This interpretation is very relevant for low-signal problems such as the genome-wide association studies (GWAS), where most signals are weak and one would want to know the sensitivity of the conditional non-independence.

\subsection{The floodgate algorithm}
\label{sec:fg-etv}
Having introduced ETV as a VIM, we now focus on designing algorithms to do inference on it. Astute readers might already recognize that \eqref{equation:fg_target} has a close relationship with the ability to distinguish the distributions
\begin{equation}
\mathcal L(X,Y,Z)\text{ and }\mathcal L(Y\mid Z)\times\mathcal L(X\mid Z)\times\mathcal L(Z).
\label{equation:2-dist}
\end{equation}
In fact, the ETV measure \eqref{equation:fg_target} exactly corresponds to the optimal error rate when classifying samples into the two populations. We state a general result in Theorem~\ref{theorem:tv} below.

\begin{theorem}
Let $\pi_0$ and $\pi_1$ be two distributions supported on the same continuous or discrete space $\Omega$. Let $\pi$ be the distribution of $(\omega,E)$, where $E\sim\Bern(a)$ and $\omega\mid E\sim\pi_E$. Then for any $f:\Omega\to[0,1]$,
\begin{equation*}
2\left(1-\e_\pi\left[\frac1aI(E=1)(1-f(\omega))+\frac{1}{1-a}I(E=0)f(\omega)\right]\right)\le\TV(\pi_1,\pi_0),
\end{equation*}
where the equality could be achieved by an optimal $f$.
\label{theorem:tv}
\end{theorem}
Armed with this observation, we can convert the task of inferring ETV into inferring the classification error rate of samples from the two distributions in \eqref{equation:2-dist}. We could collect samples from $\mathcal L(X,Y,Z)$, and the design-based/model-X approach \citep{rubin1974estimating,holland1986statistics,janson2017model,candes2018panning} allows us to obtain samples from $\mathcal L(Y\mid Z)\times\mathcal L(X\mid Z)\times\mathcal L(Z)$ by utilizing our ability to sample from $\mathcal L(X\mid Z)$. 
In randomized experiments \citep{rubin1974estimating,holland1986statistics}, $\mathcal L(X\mid Z)$ is known by design, such as the example of conjoint analysis we present in Section~\ref{sec:conjoint}. Even for some observational data sets, $\mathcal L(X\mid Z)$ can be estimated accurately from unlabeled samples of $(X, Z)$ without $Y$. This assumption can be further relaxed if one is willing to assume certain parametric models for $\mathcal L(X\mid Z)$
\citep[Section 3.2]{LZ-LJ:2020}, and the floodgate approach for other VIMs has been extended to work under doubly robust assumptions, but such an extension is beyond the scope of this paper.

Assuming it is possible to sample from $\mathcal L(X\mid Z)$, it is then straightforward to apply Theorem~\ref{theorem:tv} to this specific case.
\begin{corollary}
Let $(X^{(0)},Y,Z)\sim p(z)p(x|z)p(y|x,z)$ and $X^{(1)},\dots,X^{(J)}\mid X^{(0)},Y,Z\simiid p(x|Z)$. For any $f:(\mathcal X, \mathcal Y, \mathcal Z)\to[0,1]$, 
\begin{multline*}
2\left(1-\e\left[1-f(X^{(0)},Y,Z)+\frac{1}{J}\sum_{j=1}^Jf(X^{(j)},Y,Z)\right]\right)\\
\le\int \Big |p(y|x,z)-p(y|z)\Big| p(x|z)p(z)\di x\di y\di z.
\end{multline*}
\label{corollary:fg}
\end{corollary}
Taking advantage of Corollary~\ref{corollary:fg}, we can design an algorithm that produces a lower confidence bound of \eqref{equation:fg_target} using the central limit theorem.
\begin{algorithm}
\caption{Floodgate for categorical responses.}
\label{algorithm:cat-fg}
\textbf{Input}: An i.i.d. data set $(X_i,Y_i,Z_i)_{i=1}^n$, conditional distribution $\mathcal L(X\mid Z)$ and a classifier $f:(\mathcal X, \mathcal Y, \mathcal Z)\to[0,1]$, number of resamples $J$, $Y$'s support size $|\mathcal Y|$, confidence level $\alpha\in(0,1)$\\
\textbf{Output}: a lower confidence bound for $\ETV(X,Y,Z)$\\
	{
		\For{$i=1$ \KwTo $n$} {
        Draw $X_i^{(1)},\dots,X_i^{(J)}\simiid\mathcal L(X\mid Z=Z_i)$.\\
        Set $(Y_i^{(j)},Z_i^{(j)})=(Y_i,Z_i)$, $E_i=1$ and $E_i^{(j)}=0$, $1\le j\le J$.\\
        $L_i\leftarrow(|f(X_i,Y_i,Z_i)-1|+(1/J)\sum_{j=1}^J|f(X_i^{(j)},Y_i^{(j)},Z_i^{(j)})-0|)$.
	}
	}
    $\bar{L}\leftarrow\sum_{i=1}^nL_i/n$\\
	$\bar{L^2}\leftarrow\sum_{i=1}^nL_i^2/n$\\
	Return $L_n^\alpha(f)=\max(0, (1-\bar{L}-z_\alpha\sqrt{\bar{L^2}-\bar{L}^2}/\sqrt n)/(1-1/|\mathcal Y|))$, where $z_\alpha$ satisfies $1-\Phi(z_\alpha)=\alpha$.
\end{algorithm}
The algorithm works by first choosing a classification function $f$ as in Corollary~\ref{corollary:fg}, then using the sample mean and variance of classification error to produce a lower confidence bound for the real classification accuracy rate, which is itself a lower bound of the optimal classification accuracy rate and a re-scaled ETV. This idea of producing a lower confidence bound of a lower bound of the quantity of interest is metaphorically termed ``floodgate'' in \citet{LZ-LJ:2020}, hence the name of Algorithm~\ref{algorithm:cat-fg}. Note that there is an oracle $f$ that provides the best lower confidence bound, in the sense given in Theorem~\ref{theorem:cat-fg-valid}. Thus, the coverage of Algorithm~\ref{algorithm:cat-fg} can be tight.
\begin{theorem}[Validity of Algorithm~\ref{algorithm:cat-fg}]
For any given $f$ and $\alpha\in(0,1)$, $\lim_{n\to\infty}\p(\ETV\ge L_n^\alpha(f))\ge1-\alpha$. Additionally, $\lim_{n\to\infty}\p(\ETV\ge L_n^\alpha(f_\textnormal{oracle}))=1-\alpha$, where
\begin{equation}
f_\textnormal{oracle}(x,y,z)=\mathbb I(p(y\mid x,z)>p(y\mid z))=\mathbb I\left(\frac{p(y\mid x,z)}{p(y\mid x,z)+p(y\mid z)}>0.5\right).
\label{equation:f-oracle}
\end{equation}
\label{theorem:cat-fg-valid}
\end{theorem}
The proof of Theorem~\ref{theorem:cat-fg-valid} follows directly from Corollary~\ref{corollary:fg} and the central limit theorem once we note that $L_i$'s are bounded and independent and identically distributed. The $f_\textnormal{oracle}$ that achieves exact coverage is constructed in the proof of Corollary~\ref{corollary:fg}.

A natural question that the reader may have is whether we could provide an \emph{upper} confidence bound for the ETV. We present Theorem~\ref{theorem:impossible-upper-bound}, which states that in some sense the answer is no: a generic confidence upper bound of the ETV must simply cover the theoretical upper bound even under the most ideal scenario: no $Z$ variable, $X$ and $Y$ independent, $X$'s distribution is known, and $Y$'s distribution is uniform.
\begin{theorem}
Let $(X_i,Y_i)_{i=1}^n$ be i.i.d. samples from $\mathcal L$, where the marginal distribution of $Y_i$ is $\Unif(\{1,\dots,K\})$. Let $C_{\mathcal L_X}$ be an algorithm tailored for the marginal distribution of $X_i$ that takes $(X_i,Y_i)_{i=1}^n$ as input and produces a confidence upper bound, such that $\p_{\mathcal L}(C_{\mathcal L_X}(X_{1:n}, Y_{1:n})\ge\ETV(X,Y))\ge1-\alpha$, $\alpha\in(0,1)$, for any $\mathcal L$ that respects the marginal distributions $\mathcal L_X$, where $\ETV(X,Y)$ is \eqref{equation:fg_target} with an empty $Z$. Then,
\begin{equation}
\p(C_{\mathcal L_X}(X_{1:n}, Y_{1:n})\ge1)\ge1-\alpha
\label{equation:upper-bound}
\end{equation}
when $(X_i,Y_i)\simiid\mathcal L_X\times\Unif(\{1,\dots,K\})$, where $\mathcal L_X$ is a continuous distribution and $1$ is the theoretical ETV upper bound given by Lemma~\ref{lemma:etv-range}.
\label{theorem:impossible-upper-bound}
\end{theorem}
 To provide some intuition on Theorem~\ref{theorem:impossible-upper-bound}, we can understand the hardness of producing an upper bound by thinking about the general problem of upper bounding the total variance distance between $\mathcal L_1$ and $\mathcal L_2$. By writing $\TV(\mathcal L_1,\mathcal L_2)=\sup_A|\mathcal L_1(A)-\mathcal L_2(A)|$, we can easily obtain a lower bound for $\TV$ by fixing a non-trivial set $A$, and it is then straightforward to empirically estimate the lower bound $|\mathcal L_1(A)-\mathcal L_2(A)|$. However, in order to upper bound or estimate the actual $\TV$, one would need to be able to consistently estimate the set $A=\argmax_A|\mathcal L_1(A)-\mathcal L_2(A)|$. For the ETV, this translates to consistently estimating the optimal classifier $f$ given by \eqref{equation:f-oracle}, which requires to impose conditions on $\mathcal L(Y\mid X, Z)$. On the other hand, we do not need any such assumptions to produce a lower confidence bound.

Moving back to Algorithm~\ref{algorithm:cat-fg}, the function $f$ in practice would have to be trained on a separate dataset to maintain validity, which is not fully utilizing the whole dataset. Next, we show how to apply cross-validation in a way that every data point is used for inference. 

\paragraph{Data splitting, cross validation and derandomization}
To avoid excluding any data in the inference step, we use the idea of cross-validated floodgate from \citet{zhangfloodgate}. The idea borrows results from central limit theorems for cross-validation \citep{austern2020asymptotics,NEURIPS2020_bce9abf2} and ensures the validity of Algorithm~\ref{algorithm:cat-fg-cv}, a cross-validated version of Algorithm~\ref{algorithm:cat-fg}.
\begin{algorithm}
\caption{Cross-validated floodgate for categorical responses.}
\label{algorithm:cat-fg-cv}
\textbf{Input}: An i.i.d. data set $(X_i,Y_i,Z_i)_{i=1}^n$, conditional distribution $\mathcal L(X\mid Z)$ and a classifier training rule $f$, number of resamples $J$, number of CV folds $k$, $Y$'s support size $|\mathcal Y|$, confidence level $\alpha\in(0,1)$\\
\textbf{Output}: a lower confidence bound for $\ETV(X,Y,Z)$\\
	{Randomly partition the data into $k$ folds $B_1^c,\dots,B_k^c$ with sizes differing by at most one.
		\For{$r=1$ \KwTo $k$} {
        Train a classifier $f_{B_r}$ with data $B_r$, plug in $f_{B_r}$ and data $(X_i,Y_i,Z_i)_{i\in B_r^c}$ to Algorithm~\ref{algorithm:cat-fg}, and record the sample mean and sample variance of the $L$ vector as $\hat\mu_r$ and $\hat\sigma^2_r$.
	}
	$\hat\mu\leftarrow\sum_{r=1}^k\hat\mu_r/k$\\
	$\hat\sigma^2\leftarrow\sum_{r=1}^k\hat\sigma^2_r/k$\\
	}
	Return $L_n^\alpha(f)=\max(0,(1-\hat\mu-z_\alpha\hat\sigma/\sqrt n)/(1-1/|\mathcal Y|))$, where $z_\alpha$ satisfies $1-\Phi(z_\alpha)=\alpha$.
\end{algorithm}

\begin{theorem}[Validity of Algorithm~\ref{algorithm:cat-fg-cv}]
For any given $f$ and $\alpha\in(0,1)$, let
\begin{equation*}
\begin{aligned}
h_n((x,x^{(1:J)},y,z); B_1)&=|f_{B_1}(x,y,z)-1|+\frac1J\sum_{j=1}^J|f_{B_1}(x^{(j)},y,z)|,\\
\bar h_n((x,x^{(1:J)},y,z))&=\e_{B_1}[h_n((x,x^{(1:J)},y,z); B_1)],\\
\sigma_n&=\sqrt{\Var(\bar h_n((X,X^{(1:J)},Y,Z)))}
\end{aligned}
\end{equation*}
where the subscript $B_1$ means taking expectation over $B_1$, which contains the $n(1-1/k)$ training samples for $f_{B_1}$. Assume
\begin{enumerate}[label=(\alph*)]
    \item $\left(\bar h_n((X,X^{(1:J)},Y,Z))-\e[\bar h_n((X,X^{(1:J)},Y,Z))]\right)/\sigma_n^2$ is uniformly integrable;
    \item and the asymptotic linearity condition (2.2) in \citet{NEURIPS2020_bce9abf2} holds in probability:
    \begin{multline*}
    \frac{1}{\sigma_n\sqrt n}\sum_{r=1}^k\sum_{i\in B_r^c}\Big((h_n(X_i,X^{1:J}_i,Y_i,Z_i); B_r)-\e[h_n(X_i,X^{1:J}_i,Y_i,Z_i); B_r)\mid B_r]\\
    -\left(\bar h_n((X,X^{(1:J)},Y,Z))-\e[\bar h_n((X,X^{(1:J)},Y,Z))]\right)\Big)\cip0,
    \end{multline*}
\end{enumerate}
then $\lim_{n\to\infty}\p(\ETV\ge L_n^\alpha(f))\ge1-\alpha$. Additionally, $\lim_{n\to\infty}\p(\ETV\ge L_n^\alpha(f_\textnormal{oracle}))=1-\alpha$, where $f_\textnormal{oracle}$ is given by \eqref{equation:f-oracle}.
\label{theorem:cat-fg-cv-valid}
\end{theorem}
Assumption~(a) holds if $\bar h_n((X,X^{(1:J)},Y,Z))$ does not converge to a degenerate distribution. Section~3 in \citet{NEURIPS2020_bce9abf2} discussed some sufficient conditions of assumption~(b). Notably, 
when the number of cross-validation folds $k=O(1)$, then a sufficient condition of (b) is
\begin{equation*}
\frac{\e[\Var[h_n((X,X^{(1:J)},Y,Z); B_1)\mid (X,X^{(1:J)},Y,Z)]]}{\Var(\bar h_n((X,X^{(1:J)},Y,Z)))}\to0\text{ in probability}.
\end{equation*}
Assuming the denominator converges to a positive constant, this condition says that the out-of-sample loss is asymptotically stable over randomness of the training sample. Because $h_n$ is bounded, the conditions of Theorem~\ref{theorem:cat-fg-cv-valid} hold if there exists $f_*$ such that $f_{B_1}(x,y,z)\to f_*(x,y,z)$ in probability, uniformly for any $(x,y,z)$.

\subsection{Classification function}
\label{sec:c}
In this section, we discuss how to train the function $f$ in Algorithms~\ref{algorithm:cat-fg} and \ref{algorithm:cat-fg-cv}.

By looking at the ultimate goal of $f$, which is to predict whether $X$ is a resample or the original sample, a greedy approach is to train $f$ by regressing $E$ on $(X,Y,Z)$ using samples
\begin{equation*}
(E_i^{(j)}, (X_i^{(j)},Y_i,Z_i)), i=1,\dots,n, j=0,\dots,J,
\end{equation*}
where $X_i^{(0)}=X_i$ and $E_i^{(j)}=\mathbb I(j=0)$. However, this approach is ignoring important structural information. From the proof of Corollary~\ref{corollary:fg}, the oracle $f$ that minimizes the expected error rate is the one given in \eqref{equation:f-oracle}, which motivates the following choice of $f$ in practice
\begin{equation}
f(x,y,z;p_{\theta_1},p_{\theta_2},c)=\left\{\begin{aligned}
&1, &\hat p_i>0.5+c,\\
&0, &\hat p_i<0.5-c\\
&0.5, &|\hat p_i-0.5|\le c,
\end{aligned}\right.
\label{equation:f-c}
\end{equation}
where
\begin{equation*}
\hat p_i = \frac{p_{\hat\theta_1}(y\mid x,z)}{p_{\hat\theta_1}(y\mid x,z)+p_{\hat\theta_2}(y\mid z)},
\end{equation*}
$\hat\theta_1$ and $\hat\theta_2$ are parameter estimates of working models $p_{\theta_1}(y\mid x,z)$ and $p_{\theta_2}(y\mid z)$ and $c$ acknowledges the estimation error and gives an extra degree of freedom. The working models can be from any model family, including simple generalized linear models and fancy machine learning models. The logic behind such $f$ is that we classify the sample to the population $0$ or $1$ that has the higher estimated likelihood, but when the two likelihoods are close and we are not sure, we set it to $0.5$ and essentially discard this one sample. The parameter $c$ controls our comfort level of confidence. We will show the empirical effect of $c$ in Section~\ref{sec:c-sim}. 

\subsection{Generalization to hierarchical responses}
\label{sec:hierarchical}

In some cases, the response $Y$ may have several levels, arranged in a hierarchy. For example, a wolf is also a type of dog, which is also an animal. We can choose to relabel wolf to dog or animal to reflect the relevant level of granularity. It is then straightforward to apply Algorithms~\ref{algorithm:cat-fg} and \ref{algorithm:cat-fg-cv} to the relabeled data. We wish to raise a subtle yet crucial point that one cannot simply drop certain labels. For example, if one only cares about a feature $X$'s ability to distinguish $Y=\text{A}$ from $Y=\text{B}$, one might be tempted to simply drop all samples where $Y\not\in\{\text{A},\text{B}\}$. However, doing so would require one to be able to sample from $\mathcal L(X\mid Z, Y\in\{\text{A},\text{B}\})$ to apply Algorithms~\ref{algorithm:cat-fg} or \ref{algorithm:cat-fg-cv}, which is a different assumption from being able to sample from $\mathcal L(X\mid Z)$.

If all values of $Y$ have the same number of levels, then we can define an overall VIM by weighting all levels. Let $Y=(Y_1,\dots,Y_K)$ have $K$ hierarchy levels, where for any possible values $Y$ and $\tilde Y$, if $Y_k\ne\tilde Y_k$, then $Y_{k'}\ne\tilde Y_{k'}$ for all $k'>k$. For instance, we can let $K=3$ and $Y_1,Y_2,Y_3$ be the taxonomic ranks of family, genus and species. Next, we define a VIM at each hierarchical level $k>1$ by
\begin{equation*}
\HETV_k(X,Y,Z)=(1-1/|\mathcal Y_{1:k}|)\ETV(X,Y_{1:k},Z)-(1-1/|\mathcal Y_{1:{k-1}}|)\ETV(X,Y_{1:(k-1)},Z),
\end{equation*}
where we add back the normalizing constant to ensure $\HETV_k(X,Y,Z)\ge0$. A sufficient condition of $\HETV_k(X,Y,Z)=0$ is
\begin{equation*}
\mathcal L(Y_k \mid X, Y_{1:(k-1)}, Z)=\mathcal L(Y_k\mid Y_{1:(k-1)}, Z), \text{equivalently } \mathcal L(X \mid Y_{1:k}, Z)=\mathcal L(X\mid Y_{1:(k-1)}, Z),
\end{equation*}
so we can interpret $\HETV_k(X,Y,Z)$ as an ETV-based VIM of $X$ at hierarchical level $k$. Finally, we define an overall VIM of $X$ by aggregating $\ETV(X,Y_1,Z)$ and HETV at all others levels with a user-specified weight vector $w$:
\begin{equation*}
\HETV_w(X,Y,Z)=w_1(1-1/|\mathcal Y_1|)\ETV(X,Y_1,Z)+\sum_{k=2}^Kw_k\HETV_k(X,Y,Z).
\end{equation*}

We can then modify Algorithm~\ref{algorithm:cat-fg} to support HETV, as below.

\begin{algorithm}[h]
\caption{Floodgate for hierarchically weighted categorical responses.}
\label{algorithm:cat-fg-hier}
\textbf{Input}: An i.i.d. data set $(X_i,Y_i,Z_i)_{i=1}^n$, conditional distribution $\mathcal L(X\mid Z)$ and classifiers $f_k:(\mathcal X, \mathcal Y_{1:k}, \mathcal Z)\to[0,1]$, number of resamples $J$, confidence level $\alpha\in(0,1)$\\
\textbf{Output}: a lower confidence bound for $\ETV_w(X,Y,Z)$\\
	{
		\For{$i=1$ \KwTo $n$} {
        Draw $X_i^{(1)},\dots,X_i^{(J)}\simiid\mathcal L(X\mid Z=Z_i)$.\\
        Set $(Y_i^{(j)},Z_i^{(j)})=(Y_i,Z_i)$, $E_i=1$ and $E_i^{(j)}=0$, $1\le j\le J$.\\
        \For{$k=1$ \KwTo $K$}{
            $L_{i,k}\leftarrow (|f_k(X_i,Y_{1:k,i},Z_i)-1|+(1/J)\sum_{j=1}^J|f_k(X_i^{(j)},Y_{1:k,i}^{(j)},Z_i^{(j)})-0|)$.
        }
        $L_{i}\leftarrow w_1L_{i,1}+\sum_{k=2}^Kw_k(L_{i,k}-L_{i,k-1})$.
	}
	}
    $\bar{L}\leftarrow\sum_{i=1}^nL_i/n$\\
	$\bar{L^2}\leftarrow\sum_{i=1}^nL_i^2/n$\\
	Return $L_n^\alpha(f)=\max(0, 1-\bar{L}-z_\alpha\sqrt{\bar{L^2}-\bar{L}^2}/\sqrt n)$, where $z_\alpha$ satisfies $1-\Phi(z_\alpha)=\alpha$.
\end{algorithm}

In the same way, we could also modify Algorithm~\ref{algorithm:cat-fg-cv} to work for HETV.

\subsection{Relationship with literature}
\label{sec:comp}
Having introduced the definition of ETV and algorithms to produce its lower confidence bounds, we pause to discuss two recent works in the literature that have connections to our work.
\paragraph{Connection to MACM in \citet{LZ-LJ:2020}}
\citet{LZ-LJ:2020} defined MACM for the specific case where $Y\in\{1,-1\}$, which has exactly twice the value of ETV. Their inference \citep[Algorithm~3]{LZ-LJ:2020} is equivalent to our Algorithm~\ref{algorithm:cat-fg} with
\begin{equation}
f(X,Y,Z)=\left\{\begin{aligned}
&\mathbb I(\mu(X,Z)\ge\e[\mu(X,Z)\mid Z]),\text{ if }Y=1,\\
&\mathbb I(\mu(X,Z)\le\e[\mu(X,Z)\mid Z]),\text{ if }Y=-1.
\end{aligned}\right.
\label{equation:macm}
\end{equation}
The details are deferred to Appendix~\ref{sec:macm}.
\paragraph{Connection to $\hat\lambda^\rho_\text{bayes}$ in \citet{michel2022high}}
\citet{michel2022high} studied lower confidence bounds for $\TV(P,Q)$ based on i.i.d. samples from $P$ and $Q$. One of their proposed estimators, $\hat\lambda^\rho_\text{bayes}$ in \citet[Proposition~3]{michel2022high}, is based on the same classification idea as Algorithm~\ref{algorithm:cat-fg}. Specifically, \citet{michel2022high} also utilized the relationship between the classification accuracy and the total variation distance, and $\hat\lambda^\rho_\text{bayes}$ is constructed based on this fact for a fixed classification function $\rho_t(x)=\mathbb I(\rho(x)>t)$ with $t=0.5$. Similar to our discussion around the parameter $c$ in Section~\ref{sec:c}, \citet{michel2022high} showed that there may exist better choices for $t$ in $\rho_t(x)$ than the natural $t=0.5$, depending on prior knowledge of $P$ and $Q$. While we propose to use cross-validation to choose $c$, \citet{michel2022high} went on to consider estimators very different from $\hat\lambda^\rho_\text{bayes}$.
While our method shares the same construction idea as $\hat\lambda^\rho_\text{bayes}$ 
 in \citet{michel2022high}, the key difference between the two works is the problem setting. \citet{michel2022high} studied two-sample testing, where the samples are naturally labeled with auxiliary information; our work is centered around the ETV, which is a novel VIM defined through a sample-labeling mechanism based on $\mathcal L(X\mid Z)$. 
\section{Simulations}
\label{sec:sim}
\subsection{Floodgate with different classification functions}
\label{sec:sim-f}
In this section, we consider the model
\begin{equation}
    \begin{aligned}
        Y\mid X&\sim\Bern(\Phi(X^\top\beta)),\\
        X&\sim\mathcal N(0,\Sigma),
    \end{aligned}
\label{equation:simulation-model}
\end{equation}
where $X$ is a $p$-dimensional column vector and we provide lower confidence bound for $\ETV(X_j,Y,X_{\text-j})$ for each $j$. We choose $\Sigma_{ij}=\rho^{|i-j|}$. We set $p=4$ or $10$, $\beta=(0,1,2,3)$ for $p=4$ and $\beta=(0,0,0,0,1,2,3,4,5,6)$ for $p=10$, $n=100p$, and apply 10-fold cross validation in Algorithm~\ref{algorithm:cat-fg-cv}. We use three types of classification functions as in \eqref{equation:f-c}. For the oracle model, $p_{\theta_1}$ and $p_{\theta_2}$ are set to the true models. For the logistic or tree models, $p_{\theta_1}$ and $p_{\theta_2}$ are logistic or tree models, and $\hat\theta_1$ and $\hat\theta_2$ are trained on cross-validated data. We find that the oracle gives the highest floodgate bound (as expected), and the logistic model is a close second. Even the generic random forest model performs reasonably well.

\begin{figure}[H]
    \centering
\includegraphics[width = 1\textwidth]{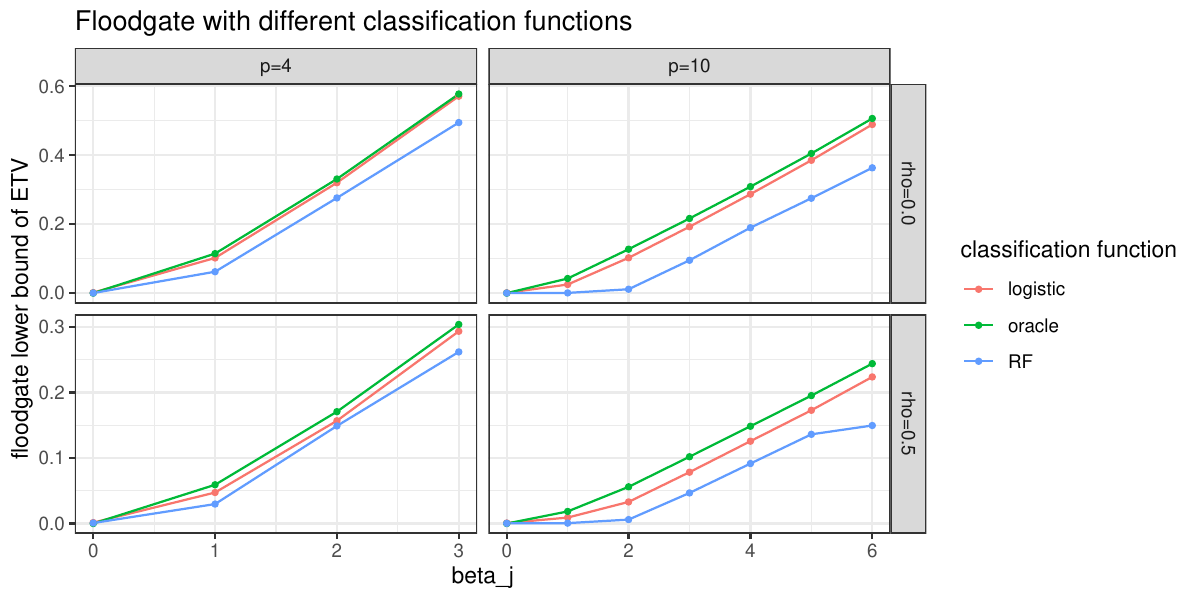}
    \caption{Floodgate lower bound with different calssification functions, averaged over 1536 independent experiments.}
    \label{figure:fg}
\end{figure}


\subsection{Effect of threshold $c$}
\label{sec:c-sim}
In this section, we demonstrate the effect of $c$ in \eqref{equation:f-c}. We consider a model of the form
\begin{equation}
    \begin{aligned}
        Y\mid X&\sim\Bern\left(\Phi\left(\beta_kX_kZ_k+\sum_{j\ne k}\beta_jX_j
        \right)\right),\\
        X&\sim\mathcal N(0,\Sigma), Z\sim\Bern(0.5), X\ci Z,
    \end{aligned}
\label{equation:simulation-model-c}
\end{equation}
where $X$ is a $p$-dimensional column vector and we provide lower confidence bound for $\ETV(X_k,Y,(X_{\text-k},Z))$, where we rotate the value of $k$. We choose $\Sigma_{ij}=\rho^{|i-j|}$. We set $p=10$, $\beta=(0,0,0,1,2,3,4,5,6,7)$ and $n=200$. We focus on two classification functions as in \eqref{equation:f-c} and the results are reported in Figure~\ref{figure:fg_c}. For the ``logistic'' model, we use logistic models for $p_{\theta_1}$ and $p_{\theta_2}$, and $\hat\theta_1$ and $\hat\theta_2$ are trained on cross-validated data; for the ``logistic\_int'' model, we add interactions between $Z$ and other $X_j$'s into the models. We explore the following methods to choose $c$.

\begin{enumerate}
    \item ``Naive'' means setting $c=0$.
    \item ``CV'' means using 10-fold cross validation to choose $c$. 
    Note that this cross validation is  different  from one we use to train $f$ in Algorithm~\ref{algorithm:cat-fg-cv}.
\end{enumerate}
We further compare our methods with an oracle method described below. Consider the $r$th fold in Algorithm~\ref{algorithm:cat-fg-cv}, where we have trained classifier $f_{B_r}$ and the evaluation set $B_r^c$. We use $\hat\mu(f_{B_r,c}, D)$ and $\hat\sigma^2(f_{B_r,c}, D)$ to denote the sample mean and variance of applying Algorithm~\ref{algorithm:cat-fg} with dataset $D$ and $f_{B_r,c}$, where $f_{B_r,c}$ is combining $c$ with $f_{B_r}$ as in \eqref{equation:f-c}. ``CV\_oracle'' means setting
\begin{equation*}
c_\text{oracle}=\argmax_c\frac1K\sum_{k=1}^K\left(1-\hat\mu(f_{B_r,c}, D_k)-z_\alpha\hat\sigma(f_{B_r,c}, D_k)/\sqrt{|B_r\cup B_r^c|}\right),
\end{equation*}
where $D_1,\dots,D_K$ are $K$ independent regeneration of dataset $B_r^c$ with the true distribution.

The results are summarized in Figure~\ref{figure:fg_c}. We can see that ``Oracle'' outperforms ``CV'' and ``Naive'', matching intuition. ``CV'' outperforms ``Naive'' and is quite close to the oracle method.
\begin{figure}[H]
    \centering
\includegraphics[width = 1\textwidth]{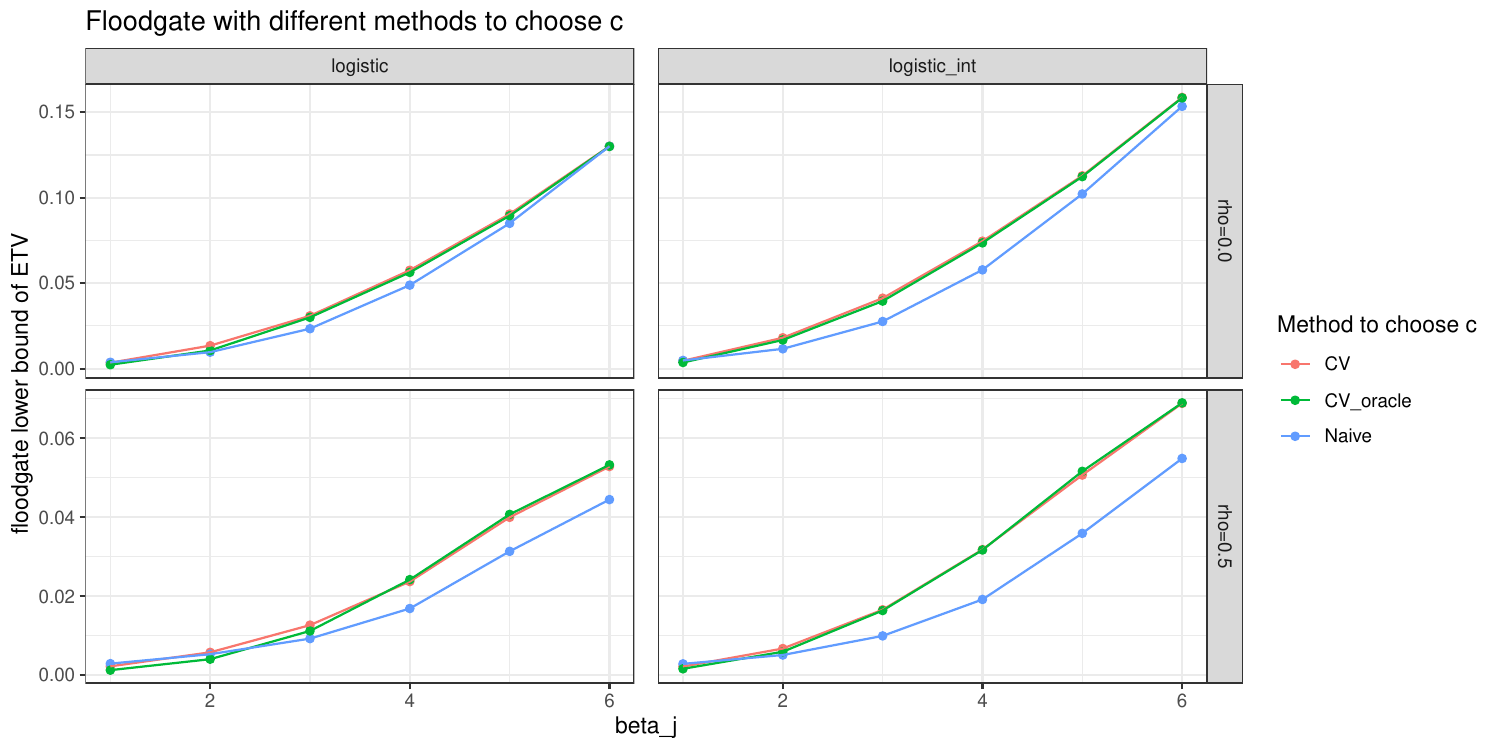}
    \caption{Violin plots of Floodgate bound based on 1536 experiments.  ``logistic'' and ``logistic\_int'' denote the classification function.}
    \label{figure:fg_c}
\end{figure}
\section{Application in conjoint analysis}
\label{sec:conjoint}
\subsection{Conjoint analysis}
Conjoint analysis \citep{luce1964simultaneous} is a survey-based statistical technique, where respondents are given a number of profiles with different attributes are asked to pick a favourite or rank them. A popular VIM used by social scientists is the average marginal treatment effect (AMCE), and there has been work on constructing confidence intervals on the AMCE \citep{hainmueller2014causal,ono2019contingent}. The AMCE, as its name suggests, considers only the marginal effect and may fail to capture some interactions. \citet{DH-KI-LJ:2022} introduced a hypothesis testing procedure for the null hypothesis $Y\ci X\mid Z$ in the conjoint analysis context, but they did not propose a VIM. We will bridge this gap by using the ETV as the VIM in conjoint analysis and construct confidence intervals on it.

\subsection{US general election data}
In this section, we analyze the election data in \citet{ono2019contingent}. In the experiment, each respondent is given two hypothetical political candidate profiles and asked to pick the one that they prefer. Each data point can thus be written the form
\begin{equation*}
(Y, X^0, X^1, Z^0, Z^1, Z^R),
\end{equation*}
where $(X^k,Z^k)$ are the attributes of Candidate $k$ with $X$ being the attribute of interest, $Z^R$ is the attribute of the respondent, and $Y\in\{0,1\}$ is the choice of the respondent. We use $Z$ to denote the collection of $(Z^1,Z^2,Z^R)$. We focus on the presidential election data with $n=7190$ observations. In each observation, there are 13 attributes of two political candidates and 11 attributes of the respondent, so $X^0,X^1$ are scalars, $Z^0,Z^1$ are 12-dimensional and $Z^R$ is 11-dimensional. Here, each candidate's attributes are uniformly and independently randomized, with a few hard constraints; for example, a candidate with a high-skill profession must have at least two years of college experience. More details on the data can be found in Appendix~\ref{sec:data-details}.



\subsection{Floodgate inference for ETV}
\label{sec:conjoint-inference}
We choose $X^0$ and $X^1$ to be the party affiliations of the candidates, which take value from \{\textbf{D}emocratic, \textbf{R}epublican\}. We can see that while one would expect $X^{0,1}$ to play an important role in the respondent's choice $Y$, its marginal effect would be close to zero (assuming there is no party affiliation bias in the respondents). To use the AMCE, we would have to re-define $X^{0,1}$ as whether that candidate has the same party affiliation as the respondent. The ETV, on the other hand, can be employed directly. This issue could be more severe for other features that are not as straightforward to correct. For instance, the original analysis in \citet{ono2019contingent} based on the AMCE dismissed gender as a statistically significant factor for congressional political candidates, while the analysis \citet{DH-KI-LJ:2022} suggested that gender does matter for congressional candidates through interactions with other factors, including the respondent's party affiliation.

We have shown in Lemma~\ref{lemma:etv-range} that in the case of binary response, the upper bound of ETV is $1$. In our specific case, we should expect even lower upper bound.

Suppose we have the following ideal data generating distribution, where
\begin{equation*}
\p(\text{candidate party affiliation is independent})=q\in[0,1],
\end{equation*}
\begin{equation*}
X=(X^0, X^1)\mid Z\sim\Unif\{(\text{D},\text{D}), (\text{D},\text{R}), (\text{R},\text{D}), (\text{R},\text{R})\},
\end{equation*}
and
\begin{equation}
Y\mid X,Z\sim\left\{\begin{aligned}
&\Bern(0.5),\text{ if respondent is independent or two candidates have same party affiliation};\\
&\Bern(p),\text{ if candidates' party affiliation differ and candidate 1 is same as respondent};\\
&\Bern(1-p),\text{ if candidates' party affiliation differ and candidate 0 is same as respondent}.\\
\end{aligned}
\right.
\label{equation:resampled-y}
\end{equation}
In this case, $\ETV(X,Y,Z)=(1-q)|p-0.5|$. In the election data, $q\approx0.27$, so even if $p=1$, which means a respondent deterministically prefers the candidate from the same party, $\ETV(X,Y,Z)$ is merely around 0.365, far from the general upper bound of $1$. Simulations show that we are able to produce floodgate gate lower bound close to the actual ETV with Algorithm~\ref{algorithm:cat-fg-cv}. The derivation and supportive simulations are included in Appendix~\ref{sec:conjoint-etv}.

Returning to the real data analysis, we apply Algorithm~\ref{algorithm:cat-fg-cv} with $k=10$ and $J=100$. The classifier family $f$ is chosen to be the model-based $f$ in equation~\eqref{equation:f-c}, where the models are HierNet \citep{bien2013lasso}, following \citet[Section~3.3]{DH-KI-LJ:2022}. 
We summarize our analysis in Figure~\ref{figure:conjoint-results}. Each violin plot summaries 40 independent runs. Here, we include both the floodgate lower bound and the floodgate estimate (that is, manually setting the confidence interval width to zero). We use the ``Naive'' and ``CV'' methods to choose $c$ as in Section~\ref{sec:c-sim}. 
We can see that activating $c$ in $f$ boosts performance, and we obtain an ETV estimate of around 0.1 and an ETV lower bound of around 0.08. The 0.01 estimate translates to around $p=0.63$ in model~\eqref{equation:resampled-y}.
Further details are deferred to Appendix~\ref{sec:conjoint-details}.

\begin{figure}[h]
    \centering
\begin{subfigure}{.45\textwidth}
  \centering
  \includegraphics[width=.9\linewidth]{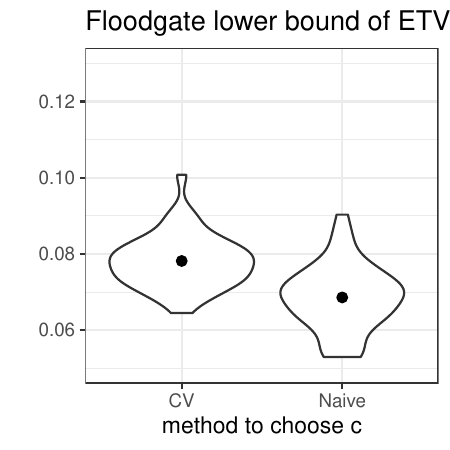}
  \caption{95\%-ETV floodgate lower bound.}
\end{subfigure}
\begin{subfigure}{.45\textwidth}
  \centering
  \includegraphics[width=.9\linewidth]{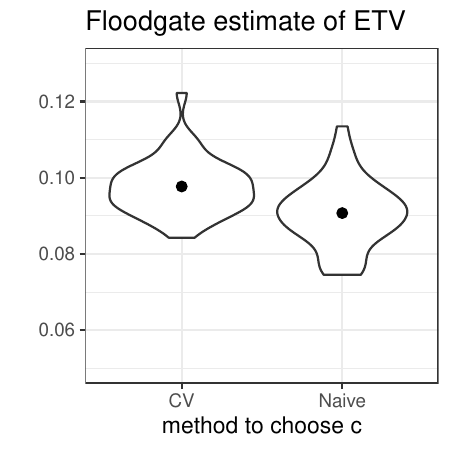}
  \caption{ETV floodgate estimate.}
\end{subfigure}
    \caption{Conjoint analysis result of the US general election data. The black dots denote the mean of the violin plots.}
    \label{figure:conjoint-results}
\end{figure}


\section*{Acknowledgements}
L.J. was partially supported by a CAREER grant from the National Science Foundation (Grant \#DMS2045981).

\bibliography{refs}{}

\newpage
\appendix
\section{Proofs}
\begin{proof}[Proof of Lemma~\ref{lemma:etv-range}]
Let $\mathcal X$, $\mathcal Y$ and $\mathcal Z$ be spaces $X$, $Y$ and $Z$ live in; let $p$, $q$, $r$ and $s$ denote the densities of $\mathcal L(Z)$, $\mathcal L(X\mid Z)$, $\mathcal L(Y\mid X, Z)$ and $\mathcal L(Y\mid Z)$. We only prove the case where $|\mathcal Y|<\infty$, while the case $|\mathcal Y|=\infty$ can be treated similarly.

When $|\mathcal Y|<\infty$, We scale \eqref{equation:fg_target} as
\begin{equation*}
\begin{aligned}
2(1-1/|\mathcal Y|)\ETV(X,Y,Z)&=\sum_{y\in\mathcal Y}|r(y\mid x,z)-s(y\mid z)|\int_{x\in\mathcal X}q(x\mid z)\di x\int_{z\in\mathcal Z}p(z)\di z\\
&=\sum_{y\in\mathcal Y}\int_{z\in\mathcal Z}\e_{X\mid Z=z}\left[|r(y\mid X,z)-s(y\mid z)|\right]p(z)\di z\\
&\le\sum_{y\in\mathcal Y}\int_{z\in\mathcal Z}2s(y\mid z)(1-s(y\mid z))p(z)\di z\text{ (Lemma~\ref{lemma:simple-1})}\\
&=2\int_{z\in\mathcal Z}\left[\sum_{y\in\mathcal Y}s(y\mid z)(1-s(y\mid z))\right]p(z)\di z\\
&\le2\int_{z\in\mathcal Z}(1-1/|\mathcal Y|)p(z)\di z=2(1-1/|\mathcal Y|)\text{ (Lemma~\ref{lemma:simple-2})}.
\end{aligned}
\end{equation*}
Here, we are using two simple lemmas of which the proofs are omitted. The upper bound is achieved when $X$, conditional on $Z$, deterministically determines $Y$ and $s(y\mid Z)=1/|\mathcal Y|$ almost surely for all $y$.
\end{proof}

\begin{lemma}
\label{lemma:simple-1}
If $X\in[0,1]$, $\e[X]=\mu$ and $\p(X=\mu)=p$, then $\e[|X-\mu|]\le2(1-p)\mu(1-\mu)$, where the equality is achieved when $X\mid(X\ne\mu)\sim\textnormal{Bern}(\mu)$.
\end{lemma}
\begin{lemma}
\label{lemma:simple-2}
Let $0\le a_i\le 1$, $\sum_{i=1}^n a_i=1$, then
\begin{equation*}
\sum_{i=1}^na_i(1-a_i)\le1-1/n.
\end{equation*}
The equality is achieved when $a_i=1/n$ for all $i$.
\end{lemma}

\begin{proof}[Proof of Theorem~\ref{theorem:tv}]
Define $\ell(f)=\e_\pi\left[\frac1aI(E=1)(1-f(\omega))+\frac{1}{1-a}I(E=0)f(\omega)\right]$. Then
\begin{equation*}
\begin{aligned}
\ell(f)&=\e_\pi\left[|f(\omega)-E|\left(\frac1aI(E=1)+\frac{1}{1-a}I(E=0)\right)\right]\\
&=\e[(1-f(\omega))/a\mid E=1]\p(E=1)+\e\left[f(\omega)/(1-a)\mid E=0\right]\p(E=0)\\
&=\int(1-f(\omega))\pi_1(\omega)\di\omega+\int f(\omega)\pi_0(\omega)\di\omega\\
&=1+\int f(\omega)(\pi_0(\omega)-\pi_1(\omega))\di\omega.
\end{aligned}
\end{equation*}
The minimum of $\ell(f)$ is attained when
\begin{equation*}
f(\omega)=f^*(\omega)=I(\pi_0(w)<\pi_1(\omega)).
\end{equation*}
It is not hard to see that $2(1-\ell(f^*))=\TV(\pi_1,\pi_0)$.
\end{proof}
\begin{proof}[Proof of Corollary~\ref{corollary:fg}]
Let $K\sim\Unif\{0,1,\dots,J\}$ independent of $(X^{(0:J)},Y,Z)$. Then we apply Theorem~\ref{theorem:tv} with $(X^{(K)},Y,Z)$ as $\omega$ and $I(K=0)$ as $E$ to get
\begin{multline*}
2\left(1-\e\left[(J+1)I(E=1)(1-f(X^{(K)},Y,Z))+\frac{J+1}{J}I(E=0)f(X^{(K)},Y,Z)\right]\right)\\
\le\int |p_{y|x,z}(y|x,z)-p_{y|z}(y|z)|p_{x|z}(x|z)p_z(z)\di x\di y\di z.
\end{multline*}
Evaluate the left hand side by conditioning on $K$ and we prove the claim.
\end{proof}

\begin{proof}[Proof of Theorem~\ref{theorem:impossible-upper-bound}]
The proof technique of this theorem is a generalization of the strategy used in the proof of \citet[Lemma~1]{barber2020distribution}, which is itself a generalization of the construction used in the proof of \citet[Proposition~5.1]{vovk2005algorithmic}.

We fix $\mathcal L_X$ in the proof, so we omit the subscript $\mathcal L_X$ of $C$. We partition the sample space of $X$ into $NK$ equal-probability Borel sets $B_{1:NK}$, $N>n$, which is possible because $\mathcal L_X$ is a continuous distribution.

We are going to define data generating distributions $D_0,\dots,D_5$ for $(X_i,Y_i)_{i=1}^n$, where $D_0$ is the distribution we care about, and we construct $D_{1:5}$ in a way such that $\TV(D_{i-1},D_{i})$ is small for $i=1,\dots,5$. Our goal is to show \eqref{equation:upper-bound} holds for $D_5$, so that it also has to hold for $D_0$. We use $\mathcal L(B)$ to denote the distribution $\mathcal L$ restricted to the set $B$.
\begin{itemize}
    \item $D_0$: sample $(X_i,Y_i)\simiid\mathcal L=\mathcal L_X\times\Unif(\{1,\dots,K\})$;
    \item $D_1$: randomly sample $n$ sets $\tilde B_{1:n}$ with replacement from $B_{1:NK}$; sample $Y_i\simiid\Unif(\{1,\dots,K\})$ and $X_i\mid\tilde B_{1:n}\sim \mathcal L_X(\tilde B_i)$ independently;
    \item $D_2$: randomly sample $n$ sets $\tilde B_{1:n}$ without replacement from $B_{1:NK}$; sample $Y_i\simiid\Unif(\{1,\dots,K\})$ and $X_i\mid\tilde B_{1:n}\sim \mathcal L_x(\tilde B_i)$ independently;
    \item $D_3$: randomly permutate $B_{1:NK}$ to be $(\tilde B_{k,m})_{1\le k\le K, 1\le m\le N}$; sample $Y_i\simiid\Unif(\{1,\dots,K\})$, sample $I_i\simiid\Unif(\{1,\dots,N\})$ but resample until all the $I_i$'s are distinct, and then sample $X_i\mid\tilde B\sim \mathcal L_x(\tilde B_{Y_i,I_i})$ independently;
    \item $D_4$: randomly permutate $B_{1:NK}$ to be $(\tilde B_{k,m})_{1\le k\le K, 1\le m\le N}$; sample $Y_i\simiid\Unif(\{1,\dots,K\})$, $I_i\simiid\Unif(1:N)$ and $X_i\mid\tilde B\sim \mathcal L_x(\tilde B_{Y_i,I_i})$ independently;
    \item $D_5$: randomly permutate $B_{1:NK}$ to be $(\tilde B_{k,m})_{1\le k\le K, 1\le m\le N}$; sample $Y_i\simiid\Unif(\{1,\dots,K\})$ and $X_i\mid Y_i,\tilde B\sim \mathcal L_X\left(\cup_{m=1}^N\tilde B_{Y_i,m}\right)$ independently;
\end{itemize}
By assumption, because $D_5$ is an i.i.d. data generating distribution for $(X_i,Y_i)$ conditional on $\tilde B$ that respects the marginal distributions of $X$ and $Y$,
\begin{equation*}
\p_{D_5}(C(X_{1:n},Y_{1:n})\ge1\mid\tilde B)\ge1-\alpha,
\end{equation*}
where $1$ is the attained ETV upper bound per the calculation in the proof of Lemma~\ref{lemma:etv-range}. After marginalizing out $\tilde B$, we have $\p_{D_5}(C(X_{1:n},Y_{1:n})\ge1)\ge1-\alpha$.

We then notice that $D_4$ and $D_5$ are actually the same data generating distribution, so \eqref{equation:upper-bound} holds under $D_4$ as well.

Now we examine the difference between $D_3$ and $D_4$. The probability of not having to resample is $N!/(N^n(N-n)!)$, so the total variation distance between $D_3$ and $D_4$ is upper bounded by $\epsilon(n,N)=1-N!/(N^n(N-n)!)$. Thus,
\begin{equation}
\p_{D_3}(C(X_{1:n},Y_{1:n})\ge1)\ge1-\alpha-\varepsilon(n,N).
\label{equation:upper-bound-relaxed}
\end{equation}
Next, we notice that $D_2$ and $D_3$ are also the same. This is because they both essentially use $n$ random samples without replacement from $B_{1:NK}$. Therefore, \eqref{equation:upper-bound-relaxed} also holds for $D_2$.

Similarly, we can observe that the total variation distance between $D_1$ and $D_2$ is upper bounded by one minus the probability of all sampled sets $B_{1:n}$ in $D_1$ are distinct. This gives us the upper bound of $\varepsilon(n,NK)$. As a result, we get
\begin{equation}
\p_{D_1}(C(X_{1:n},Y_{1:n})\ge1)\ge1-\alpha-\varepsilon(n,N)-\varepsilon(n,NK).
\label{equation:upper-bound-relaxed2}
\end{equation}
Finally, there is no difference between $D_1$ and $D_0$, so \eqref{equation:upper-bound-relaxed2} also holds for $D_0$. Since $\epsilon(n,N)\to0$ as $N\to\infty$, the fact that \eqref{equation:upper-bound-relaxed2} holds for $D_0$ for any $N$ means that \eqref{equation:upper-bound} holds for $D_0$, as desired.
\end{proof}
\section{ETV and sensitivity analysis}
\label{sec:sens}
Let $B(X,Z,U)$ be the almost sure supremum of
\begin{equation*}
\max\left\{\frac{p(X \mid U, Z)}{p(X \mid Z)},\frac{p(X \mid Z)}{p(X \mid U, Z)}\right\}.
\end{equation*}
Then
\begin{equation*}
\begin{aligned}
p(y\mid z,x)=&\int p(y\mid z,x,u)p(u\mid z,x)\di u\\
&=\int p(y\mid z,u)\frac{p(x\mid u,z)p(u\mid z)}{p(x\mid z)}\di u\\
&=\int\underbrace{\frac{p(x\mid u,z)}{p(x\mid z)}}_{\in[1/B,B]}\underbrace{p(y\mid z,u)p(u\mid z)}_{\text{integrates to $p(y\mid z)$}}\di u\in[p(y\mid z)/B,Bp(y\mid z)].
\end{aligned}
\end{equation*}
Then
\begin{equation*}
\begin{aligned}
2(1-1/|\mathcal Y|)\ETV(X,Y,Z)&=\int|p(y\mid x,z)-p(y\mid z)|\int p(x\mid z)\di x\int p(z)\di z\\
&\le\int\max(B-1,1-1/B)p(y\mid z)\int p(x\mid z)\di x\int p(z)\di z\\
&=\max(B-1,1-1/B)=B-1.
\end{aligned}
\end{equation*}
Thus, $B(X,Z,U)\ge1+2(1-1/|\mathcal Y|)\ETV(X,Y,Z)$.
\section{Comparison with MACM}
\label{sec:macm}
Continuing equation~\eqref{equation:macm}, the $R_i$ in \citet[Algorithm~3]{LZ-LJ:2020} is equivalent to $1-L_i$ in Algorithm~\ref{algorithm:cat-fg}. Note that
\begin{equation*}
R_i=\left\{\begin{aligned}
&\p(U_i<0\mid Z_i)-\mathbb I(U_i<0),\text{ if }Y=1,\\
&\p(U_i>0\mid Z_i)-\mathbb I(U_i>0),\text{ if }Y=-1.
\end{aligned}\right.
\end{equation*}
and
\begin{equation*}
\begin{aligned}
1-L_i&=f(X_i,Y_i,Z_i)-\frac1J\sum_{j=1}^Jf(X_i^{(j)},Y_i^{(j)},Z_i^{(j)})\\
&=f(X_i,Y_i,Z_i)-\hat{\e}_{X\mid Z=Z_i}[f(X,Y_i,Z_i)]\\
&=\left\{\begin{aligned}\mathbb I(U_i\ge0)-\hat{\p}_{X\mid Z=Z_i}(U_i\ge0),\text{ if }Y_i=1,\\
\mathbb I(U_i\le0)-\hat{\p}_{X\mid Z=Z_i}(U_i\le0),\text{ if }Y_i=-1.
\end{aligned}\right.
\end{aligned}
\end{equation*}
\section{Conjoint analysis further details}
\label{sec:conjoint-app}
\subsection{Additional details about data}
\label{sec:data-details}
In this section, we include some additional details on the data used in Section~\ref{sec:conjoint}. Table~\ref{table:candidate-attr} includes attributes of the candidate profiles. Table~\ref{table:respondent-attr} includes attributes of the respondents.

\begin{table}[h]
\centering
\begin{tabular}{| m{15em} | m{25em} |}
\hline
\textbf{Attributes} & \textbf{Values}\\
\hline
Sex  &  
      Male, 
      Female
    \\ 
\hline
Age  &  
      36, 
      44, 
      52, 
      60, 
      68, 
      76
   \\
\hline
Race/Ethnicity  & 
      White, 
      Black, 
      Hispanic, 
      Asian American
   \\
\hline
Family &   
      Single (never married), 
      Single (divorced), 
      Married (no child), 
      Married (two children)
   \\
\hline
Experience in public office  &  
      12 years, 
      8 years, 
      4 years, 
      No experience
   \\
\hline
Salient personal characteristics  & 
      Provides strong leadership, 
      Really cares about people like you, 
      Honest, 
      Knowledgeable, 
      Compassionate, 
      Intelligent
   \\
\hline
Party affiliation &  
       Democrat Party, Republican Party\\
\hline
Policy area of expertise &  
      Foreign policy, 
      Public safety (crime), 
      Economic policy, 
      Health care, 
      Education, 
      Environmental issues
     \\
\hline
Position on national security  &  
      Wants to cut military budget and keep U.S. out of war, 
      Wants to maintain strong defense and increase U.S. influence
    \\ 
\hline
Position on immigrants  & 
Favors giving citizenship or guest worker status to
undocumented immigrants,
Opposes giving citizenship or guest worker status to
undocumented immigrants 
   \\
\hline
Position on abortion & 
Abortion is a private matter (pro-choice),
Abortion is not a private matter (pro-life),
No opinion (neutral) 
   \\
\hline
Position on government deficit &   
Wants to reduce the deficit through tax increase,
Wants to reduce the deficit through spending cuts,
Does not want to reduce the deficit now 
   \\
\hline
Favorability rating among public &  
      34\%, 
      43\%, 
      52\%, 
      61\%, 
      70\%
   \\
\hline
\end{tabular}
\caption{Types of attributes varied in candidate profiles (Table~1 in \citet{ono2019contingent}).}
\label{table:candidate-attr}
\end{table}

\begin{table}[h]
\centering
\begin{tabular}{| m{15em} | m{25em} |}
\hline
\textbf{Attributes} & \textbf{Values}\\
\hline
Sex  &  
      Male, 
      Female
    \\ 
    \hline
Education level & BA degree, No BA degree\\
\hline
Age group & 18-29, 30-50, 51-65, 66 or older\\
\hline
Age  &  Age in years
   \\
\hline
Social class  & 
      Lower class, Middle class, Upper class
   \\
\hline
Region &   
      South, Nonsouth
   \\
\hline
Race/Ethnicity  &  
      White, Black, Hispanic, Other
   \\
\hline
Partisanship &  
       Democrat Party, Republican Party, Independent\\
\hline
Thought on Hillary Clinton &
Dislike, Like, Neutral\\
\hline
Interest in politics & Not at all interested, Not very interested, Somewhat interested, Very interested\\
\hline
Political ideology & Conservative or liberal levels (7 levels)\\
\hline
\end{tabular}
\caption{Types of attributes recorded in respondents.}
\label{table:respondent-attr}
\end{table}

\subsection{ETV upper bound}
\label{sec:conjoint-etv}

We derive the ETV upper bound in Section~\ref{sec:conjoint-inference}. First, we notice that due to the symmetry of labeling, $P(Y=0\mid Z=z)=P(Y=1\mid Z=z)=0.5$ for any $z$. If $z$ is such that the respondent's party affiliation is independent, then $P(Y=0\mid X^0=x_0,X^1=x_1,Z=z)=P(Y=1\mid X^0=x_0,X^1=x_1,Z=z)=0.5$ for any $(x_0,x_1)$; otherwise, $P(Y=0\mid X^0=x_0,X^1=x_1,Z=z)$ takes value $0.5, 0.5, p, 1-p$ for $(x_0,x_1)\in\{(\text{D},\text{D}), (\text{D},\text{R}), (\text{R},\text{D}), (\text{R},\text{R})\}$. Then the ETV is
\begin{equation*}
\begin{aligned}
\ETV&=q\times0+(1-q)\sum_{y\in\{0,1\}}\sum_{x_0\in\{\text{R},\text{D}\}}\sum_{x_1\in\{\text{R},\text{D}\}}\frac14|P(Y=y\mid X^0=x_0,X^1=x_1,Z=z)-P(Y=y\mid Z=z)|\\
&=(1-q)\sum_{y\in\{0,1\}}\frac14\left(0+0+|p-0.5|+|1-p-0.5|\right)\\
&=(1-q)\times2\times\frac14\times|2p-1|=(1-q)|p-0.5|.
\end{aligned}
\end{equation*}

To test how well our algorithm does in this ideal setting, We regenerate synthetic $Y$ according to \eqref{equation:resampled-y} and apply Algorithm~\ref{algorithm:cat-fg-cv}. In Figure~\ref{figure:conjoint_sim}, we plot the average floodgate lower bound from 40 independent experiments (but they share the same synthetic response) and the theoretical upper bound $(1-q)|p-0.5|$. We can see that in moderate to high signal regimes, the floodgate bound is close to the theoretical upper bound.
\begin{figure}[H]
    \centering
\includegraphics[width = 0.9\textwidth]{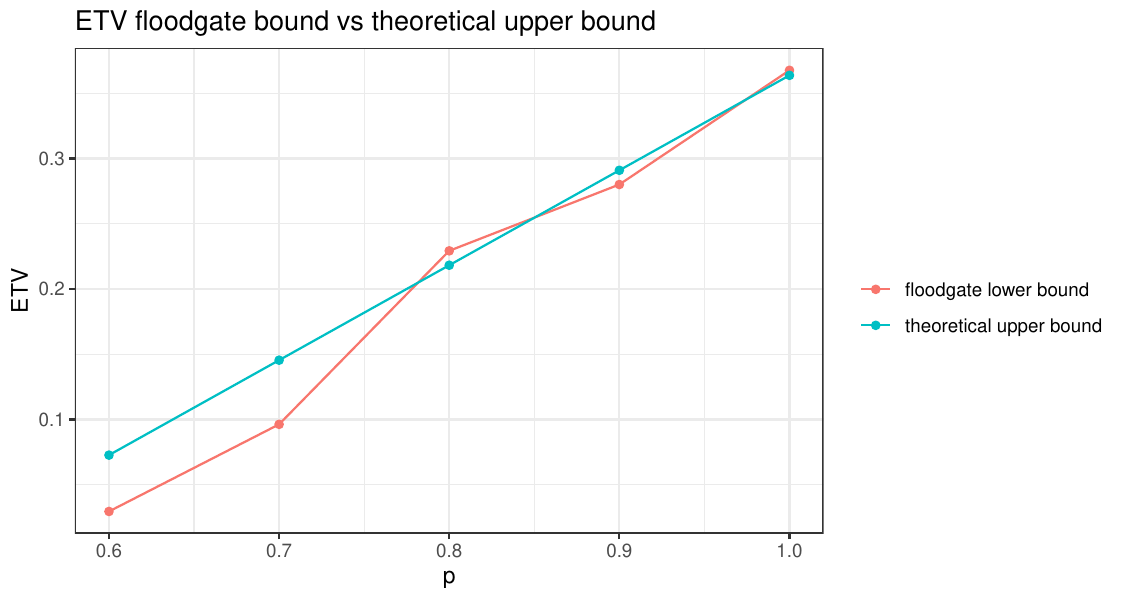}
    \caption{Conjoint analysis with synthetic responses.}
    \label{figure:conjoint_sim}
\end{figure}

\subsection{Analsyis details}
\label{sec:conjoint-details}
In the experiments in Section~\ref{sec:conjoint-inference}, $f$ is chosen to be
\begin{equation*}
f(x,y,z;p_{\theta_1},p_{\theta_2},c)=\left\{\begin{aligned}
&1, &\hat p_i>0.5+c,\\
&0, &\hat p_i<0.5-c\\
&0.5, &|\hat p_i-0.5|\le c,
\end{aligned}\right.
\end{equation*}
where $p_{\theta_1}(y\mid x,z)$ is a HierNet model with a fixed penalty parameter, where interactions between politician's gender and party affliation are added as a feature, and $p_{\theta_2}(y\mid z)$ is a HierNet model with the same penalty parameter. In the method ``CV'' to choose $c$, we further partition the training data $B_r$ into $m=10$ folds $C_{r1}^c,\dots,C_{rm}^c$. We then calculate the cross-validated loss
\begin{equation*}
\text{loss}_{r}(c)=\frac1m\sum_{j=1}^m\text{loss of }f(p_{\hat\theta_1(C_{rj})},p_{\hat\theta_2(C_{rj})},c)\text{ on }B_r\setminus C_{rj},
\end{equation*}
where $\hat\theta(C)$ means $\hat\theta$ estimated on dataset $C$, choose the $c_r$ that minimizes $\text{loss}_{r}(c)$, and let
\begin{equation*}
f_{B_r}=f(p_{\hat\theta_1(B_r)},p_{\hat\theta_2(B_r)},c_r).
\end{equation*}
\end{document}